\documentclass{article}

\usepackage{graphicx}
\usepackage{amsfonts, amsmath, amssymb, amsthm}

\usepackage[numbers,sort&compress]{natbib}
\bibpunct[, ]{[}{]}{,}{n}{,}{,}

\theoremstyle{plain}
\newtheorem{theorem}{Theorem}[section]
\newtheorem{lemma}[theorem]{Lemma}

\newtheorem{proposition}[theorem]{Proposition}

\theoremstyle{definition}
\newtheorem{definition}[theorem]{Definition}

\theoremstyle{remark}

\newcommand{\vect}[1]{\boldsymbol{#1}}

\usepackage{xcolor}

\providecommand{\detail}[1]{#1}
\renewcommand{\detail}[1]{}

\begin{document}

\title{Covariant Lyapunov vectors as global solutions of a partial
  differential equation on the phase space}

\author{
  Massimo Marino\textsuperscript{a} and Doriano
    Brogioli\textsuperscript{b,*}}

\date{\begin{flushleft}
\textsuperscript{a}Liceo Scientifico Statale ``A. Volta'', via
Benedetto Marcello 7, 20124 Milan, Italy
\\ \vspace{0.5cm}
\textsuperscript{b}Universit\"at Bremen, Energiespeicher-- und
Energiewandlersysteme, Bibliothekstra{\ss}e 1, 28359 Bremen,
Germany
\\ \vspace{0.5cm}
* brogioli@uni-bremen.de
\end{flushleft}}

\maketitle

\begin{abstract}
As a new tool to describe the behaviour of a dynamical system, we
introduce the concept of ``covariant Lyapunov field'', i.e.\ a field
which assigns all the components of covariant Lyapunov vectors at
almost all points of the phase space.  We focus on the case in which
these fields are overall continuous and also differentiable along
individual trajectories. We show that in ergodic systems such fields
can be characterized as the global solutions of a differential
equation on the phase space. Due to the arbitrariness in the choice of
a multiplicative scalar factor for the Lyapunov vector at each point
of the phase space, this differential equation exhibits a gauge
invariance that is formally analogous to that of quantum
electrodynamics. Under the hypothesis that the covariant Lyapunov
field is overall differentiable, we give a geometric interpretation of
our result: each 2-dimensional foliation of the space that contains
whole trajectories is univocally associated with a Lyapunov exponent,
and the corresponding covariant Lyapunov field is one of the
generators of the foliation.  In order to show with an example how
this new approach can be applied to the study of concrete dynamical
systems, we display an explicit solution of the differential equations
that we have obtained for the covariant Lyapunov fields in a model
involving a geodesic flow.
\end{abstract}

\vspace{0.5cm}


\vspace{0.5cm}

\section{Introduction}

Lyapunov exponents (LE) quantify the rate of divergence of
trajectories in a dynamical system~\cite{pikovsky}.  LEs give a deep
characterization of dynamical systems and have proved to be an
invaluable tool for the analysis of chaotic systems and attractors,
either in numerical calculations or experimental data~\cite{wolf1985}.
In particular, they provide the Kolmogorov-Sinai
entropy~\cite{pesin1977} and the Kaplan-Yorke dimension of an
attractor~\cite{frederickson1983}, which is an upper bound for the
information dimension of the system.

The divergence rate of trajectories depends both on the trajectory and
on the starting displacement. Different initial displacement vectors
give rise to the different observed LEs. 
Numerical procedures to calculate the LE have been known for a long
time~\cite{benettin1980, shimada1979}, and these procedures provide
as a by-product also a set of displacement
vectors, each associated with a different LE.
However, such vectors generally depend on
the chosen metric, thus they are not a characteristic of the
dynamical system. An intrinsic characterization of the system is
instead given by a suitable choice of such vectors, which are called
``covariant Lyapunov vectors'' (CLVs)~\cite{pikovsky, kuptsov2012}.

Various methods have recently been developed with the aim of
numerically calculating the CLVs; a discussion can be found in
Ref.~\cite{ginelli2013}.  \detail{``Static'' methods are summarized in
  section 6, the paper is mainly on the ``dynamic'' method.} They are
divided into the so-called ``static''~\cite{wolfe2007, kuptsov2012}
\detail{they require the calculation of the FOLVs and BOLVs}
\detail{Eq.~(41) of Kuptsov-Parlitz} and ``dynamic''
methods~\cite{ginelli2007}.  \detail{they aim at directly finding the
  CLVs} With the aid of such methods, the CLVs have been evaluated and
used as a diagnostic tool in various systems, e.g.\ in spatially
extended dynamical systems exhibiting chaos~\cite{szendro2007}, with
hyperbolic chaotic dynamics~\cite{kuptsov2009}, in large chaotic
systems consisting of globally coupled maps~\cite{takeuchi2009}, in
stationary systems out of equilibrium~\cite{bosetti2010}, and in the
phase synchronization transition of chaotic
oscillators~\cite{ouchi2011}.

Since divergence rates are calculated by
using the linearized dynamics, LEs actually depend on the direction of
the initial displacement vector and not on its norm. For this reason,
in the mathematically-oriented
literature each LE is associated, rather than to a single CLV,
to a one-dimensional subspace (or multi-dimensional subspace
in case of degeneracy) of the tangent space at each point of the phase
space.
 In this paper we adopt a new point of view, and show that the concept of CLV
 can be useful not only to perform numerical calculations, but
also within the framework of a formal mathematical treatment of the
 subject. To this purpose, we
will introduce in the next section a new definition of CLVs
which does not aim at providing a recipe for their calculation, but rather
highlights in a simple way their characteristic mathematical property.
Despite its formal novelty, such a definition is of course
fully consistent with the practical use that has
been made of CLVs in the existing literature. 
Starting from this definition, in
section \ref{sect:CLF} we define the fundamental concept of our new
approach, that of ``covariant Lyapunov field'' (CLF), i.e.\ a vector
field whose value represents a CLV, corresponding to a given LE, at
all points of its domain, which can in general be assumed to be a
subset of the phase space.  

We focus on systems in which there exist
CLFs which are continuous on their domain and differentiable along
individual trajectories. We show that under these hypotheses they can
be characterized as the global solutions of a differential equation on
the phase space.  We discuss the main properties of this equation and
we show how it leads to a geometrical interpretation of the role of
CLVs in a dynamical system.  In order to provide an explicit example
of application of our results, we show that the equation we have
obtained is actually satisfied be the CLFs in the Hadamard-Gutzwiller
model and allows us to calculate the CLVs by means of a simple
symbolic calculation.

\section{Preliminaries}

In this section, we summarize the fundamental knowledge on the topic
available in the literature, starting with the definition of Lyapunov
exponents (LEs)~\cite{pikovsky}.  Given a vector field
$\vect{F}(\vect{x})$ on an $n$-dimensional Riemannian manifold $X$,
let us consider a trajectory $\vect{x}(t)$, satisfying the
differential equation
\begin{equation}
  \frac{\mathrm{d}}{\mathrm{d}t} x^{\mu}(t) = 
  F^{\mu}\left[\vect{x}(t)\right] 
  \label{eq:xp}
\end{equation}
for $\mu=1,\dots, n$.
Then for a slightly displaced trajectory $\vect{x}(t)+\vect{\delta x}(t)$
we get at first order in $\vect{\delta x}(t)$ the equation
\begin{equation}
  \frac{\mathrm{d}}{\mathrm{d}t} \delta x^{\mu}(t) = 
  \partial_\nu F^{\mu}\left[\vect{x}(t)\right] \delta x^{\nu}(t) \, ,
  \label{eq:dxp}
\end{equation}
where $\partial_j$ is the partial derivative with respect to the
$j$-th coordinate and the terms are summed over the repeated Greek indices
(Einstein notation). 

It is typically observed that
the norm of $\vect{\delta x}(t)$,
asymptotically for $t\to +\infty$, behaves exponentially:
\begin{equation}
 \lVert \vect{\delta x}(t) \rVert =
 e^{\lambda^+ t + o(t)}  \, ,
   \label{eq:approx}
\end{equation}
where $o(t)/t$ vanishes in the limit $t\to +\infty$. The real parameter
\[
\lambda^+= \lim_{t\to +\infty} \frac{\ln \lVert \vect{\delta x}(t) \rVert}{t}
\]
is called forward Lyapunov exponent of the 
displacement vector $\vect{\delta x}_0=\vect{\delta x}(0)$
at the point $\vect{x}_0=\vect{x}(0)$. Under
quite general hypotheses it can be proved  that the value of
the LE does not depend on the particular choice of the metric tensor
used on the manifold $X$ for the calculation of the norm appearing in
Eq.~(\ref{eq:approx}).

If $\mu$ is a measure on $X$ which is preserved by the flow generated by
$\vect{F}(\vect{x})$, Oseledec's theorem \cite{Oseledec1968} says that,
at almost all points
$\vect{x}_0\in X$ with respect to the measure $\mu$,
for any tangent vector $\vect{\delta x}_0$ there exists a real number
$\lambda^+$ for which Eq.~(\ref{eq:approx}) holds.
It is easy to see that, given $\vect{x}_0$, there can be at most $n$
distinct forward LEs $\lambda_j^+$ as a function of $\vect{\delta
  x}_0$, thanks to the linearity of Eq.~(\ref{eq:dxp}) in
$\vect{\delta x}(t)$. For a rigorous discussion of the
conditions for the existence of the LEs we refer the readers to
Ref.~\cite{eckmann1985}.

The whole set of forward LEs $\lambda_j^+$, for a given $\vect{x}_0$,
can be calculated by suitable numerical procedures~\cite{benettin1980,
  shimada1979} which, as a by-product, also return a set of $n$
so-called forward ``orthonormal Lyapunov vectors'' (OLVs): each
forward LE $\lambda_j^+$ is obtained by taking one of the forward OLVs
as initial displacement vector $\vect{\delta x}_0$. As the name
suggests, the OLVs form an orthonormal set with respect to the chosen
metric tensor on $X$. However, at variance with the LEs, these OLVs do
depend on the arbitrary choice of such a metric tensor, and for this
reason they do not represent an intrinsic characterization of the
dynamical system.

A different point of view arises when the forward dynamics is compared
with the backward dynamics~\cite{pikovsky}. \detail{page 56} In
analogy with Eq.~(\ref{eq:approx}), an exponential behaviour of the
displacement vector is also observed looking at the evolution back in time,
for $t\to -\infty$:
\begin{equation}
  \lVert \vect{\delta x}(t) \rVert =
e^{\lambda^- |t|+o(t)} \, ,
  \label{eq:approx-}
\end{equation}
where $\lambda^-$ represents the backward LE and $o(t)/t$ vanishes for
$t\to -\infty$.  As for the forward LEs $\lambda_j^+$, there can be at
most $n$ distinct backward LEs $\lambda_j^-$ as a function of
$\vect{\delta x}_0$.  In general, there is no relation between the
forward and backward LEs.  However, in this work we are considering
systems with a preserved measure $\mu$; then, for almost every initial
position $\vect{x}_0$ with respect to $\mu$, \detail{theorem 3.1, page
  283} the forward and backward LEs, $\lambda_j^+$ and $\lambda_j^-$,
are opposites, and there exist initial displacement vectors
$\vect{\delta x}_0$ giving rise to these opposite LEs
~\cite{ruelle1979}. This implies that
\begin{equation}
  \lVert \vect{\delta x}(t) \rVert =
e^{\lambda_j t + o(t)} \, ,
  \label{eq:approx0}
\end{equation}
where  
\begin{equation}\label{covlyexp}
\lambda_j :=  \lambda_j^+ = -\lambda_j^- 
\end{equation}
and $o(t)/t$ vanishes for both $t\to +\infty$ and $t\to -\infty$.

The displacement vectors giving rise to Eq.~(\ref{eq:approx0}) are called
covariant (or characteristic) Lyapunov vectors with LEs $\lambda_j$
defined by Eq.~(\ref{covlyexp}).  \detail{theorem 3.1,
  page 283} \detail{also in \cite{pikovsky}, eq. 4.4 page 57 and
  \cite{kuptsov2012} Eq. 39}  Consistently with these results we will
adopt the following general definition.

\begin{definition}[{\bf covariant Lyapunov vector}]\label{def:clv}
  Let $\vect{\delta x}(t)$ be the solution of Eq.~(\ref{eq:dxp}) with
  initial data $\vect{\delta x}(0)=\vect{v}$, where $\vect{v}$ is a
  tangent vector at a point $\vect{x}_0\in X$. We say that $\vect{v}$
  is a ``covariant Lyapunov vector'' (CLV) at the point $\vect{x}_0$ if
  \begin{equation}\label{eq:limits}
    \lim_{t\to -\infty} \frac{\ln \lVert \vect{\delta x}(t) \rVert}{t} =
    \lim_{t\to +\infty} \frac{\ln \lVert \vect{\delta x}(t) \rVert}{t} \, .
  \end{equation}
  The common value $\lambda$ of the two above limits is called the
  ``Lyapunov exponent of the CLV $\vect{v}$''.
\end{definition}

According to the above definition, at a given point $\vect{x}_0$ of 
the phase space,
the CLVs corresponding to each $\lambda_j$ form a linear space of
dimension $\nu_j$, called the multiplicity of the LE $\lambda_j$, and
the sum of the multiplicities of all the LEs equals the dimension $n$
of the space $X$.  The linear space corresponding to each LE
$\lambda_j$ is independent of the particular metric which is used on
the space $X$. Hence these spaces provide a splitting of the tangent
space at almost every point of $X$ (the so called ``Oseledec
splitting'') which represents an intrinsic characterization of the
dynamical system.  The possible presence of multiplicities larger than
1, called degeneration, is often neglected in the literature, e.g.\ in
Ref.~\cite{ginelli2007}. \detail{see note 10} In the absence of
degeneration, one might say that
there is a single CLV $\vect{v}_j$, defined up to an
arbitrary scalar factor, for each LE $\lambda_j$ with $1\le j\le n$.

It is easy to see that, if $\vect{v}=\vect{\delta x}_0$ is a CLV at a
point $\vect{x}_0$ with LE $\lambda$, then the vector $\vect{\delta
  x}(t)$, evolving from $\vect{\delta x}_0$ according to
Eqs.~(\ref{eq:xp}) and (\ref{eq:dxp}), is for any $t$ a CLV at the
point $\vect{x}(t)$ with the same LE $\lambda$.  This property is
expressed by saying that the CLVs are ``invariant under the linearized
flow''~\cite{wolfe2007}, or that ``CLVs are mapped to other CLVs by
the linear propagator along trajectories''~\cite{noethen_thesis}.
This fact implies that, for a given real number $\lambda$, the set
$D\subseteq X$, of all the points at which $\lambda$ is a LE, is
invariant under the evolution of the system.

\section{Definition of covariant Lyapunov fields}\label{sect:CLF}

In this section we will introduce the covariant Lyapunov fields which, as we
already mentioned in the Introduction, are
the fundamental mathematical entity that constitutes
the original subject of our investigation.

Let us suppose that, 
for a given LE $\lambda$, a particular CLV $\vect{v}(\vect x)$
has been fixed in the tangent space at every point $\vect x$ of the
invariant set $D\subseteq X$ in which $\lambda$ is a LE.
In the absence of degeneration, fixing such a CLV 
at a point $\vect x \in D$ amounts 
to choosing a vector with a suitable norm inside the one-dimensional
subspace associated with $\lambda$. Once a CLV
has been fixed at each point of $D$, we have obtained
a vector field $\vect{v}$ on $D$ associated with $\lambda$. The aim 
of this paper is to investigate some general properties of such a 
vector field. 

\begin{figure*}
  \includegraphics{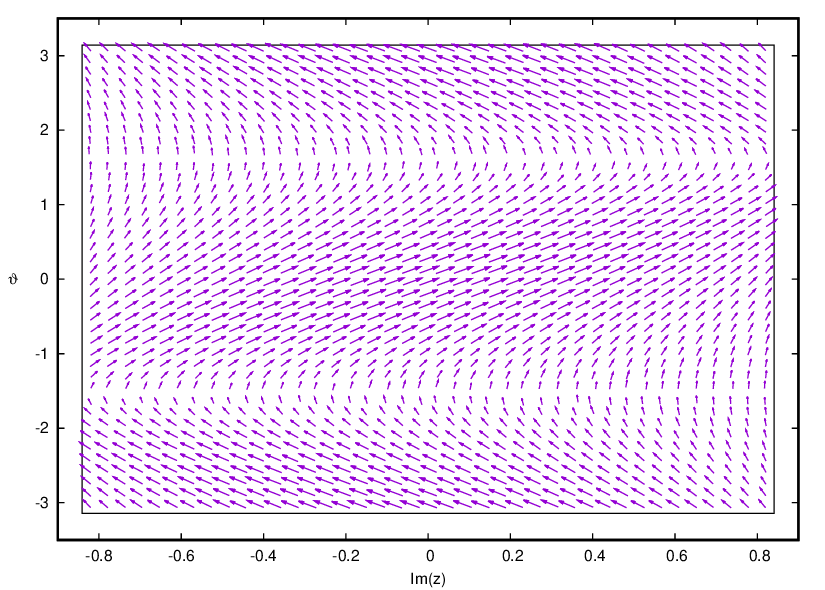}
  \caption{Covariant Lyapunov vectors of the geodesic flow on a genus-2
    hyperbolic surface of constant negative curvature. The graph shows the section
    at $\Re{z}=0$, with coordinates $\Im{z}$, $\vartheta$.
    The arrows are the projections on the section of the
    covariant Lyapunov vectors with maximum exponent (normalized with
    respect to the Euclidean norm). }
  \label{fig:hyperbolic}
\end{figure*}

We focus on continuous fields. In order to get an intuitive idea of
the continuity of these fields, in Fig.~\ref{fig:hyperbolic} we report
an example referring to the geodesic flow on a genus-2
hyperbolic surface of constant negative curvature (Hadamard-Gutzwiller
model)~\cite{braun2002}, an example of Anosov flow. It is constructed
by taking a regular hyperbolic octagon inside the Poincar\'e disc
and identifying its opposite sides. The dynamic variables are $z$,
representing the position on the complex plane, and $\vartheta$, the
angle formed by the tangent to the geodesic and the real axis. Further
detail on this model is given in section~\ref{sect:hyperbolic}.

To obtain the graph, the evolution of $\vect{x}=(\Re{z}, \Im{z},
\vartheta)$ was calculated by means of numerical integration.  The CLV
was calculated by starting with an arbitrary $\vect{\delta x}_0$ at a
point $\vect x_0$, and letting it evolve according to Eq.~(\ref{eq:dxp})
for a long time: the vector eventually approaches the CLV with the
largest LE $\lambda$. In Fig.~\ref{fig:hyperbolic} we show 
the projection of this CLV,
normalized with respect to the Euclidean norm, on the section
at $\Re{z}=0$. As one can see, the behaviour of the field appears to
be continuous everywhere on the section.

\begin{figure*}
  \includegraphics{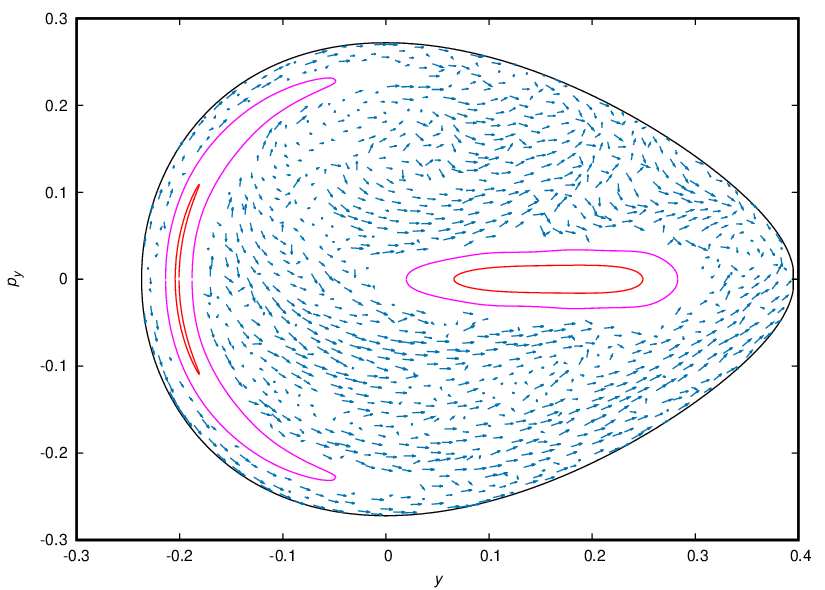}
  \caption{Covariant Lyapunov vectors of the H\'enon-Heiles system.
    The graph shows the Poincar\'e section at $x=0$ with coordinates
    $y$, $p_y$, for energy $H=0.037$ and $\Lambda=2$. The black line
    is the contour of the constant energy surface. The magenta and red
    lines are invariant tori. The arrows are the projections on the
    section of the covariant Lyapunov vectors with maximum exponent
    (normalized with respect to the Euclidean norm). }
  \label{fig:henon:heiles}
\end{figure*}

The continuity of the CLVs in other dynamical systems is more
questionable. As an example, we report in Fig.~\ref{fig:henon:heiles}
the CLV with the largest LE $\lambda$ on the Poincar\'e section of the
Henon-Heiles system~\cite{henon1964}, having the following
Hamiltonian:
\begin{equation}
  H=\frac{1}{2}\left(p_x^2+p_y^2\right) +
    \frac{1}{2}\left(x^2+y^2\right) +
    \Lambda\left(x^2y-\frac{y^3}{3}\right)\,.
\end{equation}
The CLV was calculated as in the case of the Anosov flow above.  The
CLV are not calculated in the regions covered by invariant tori;
moreover, we know that the CLV cannot be defined on the homoclinic and
heteroclinic points. In spite of these facts, in wide regions of the
graph in Fig.~\ref{fig:henon:heiles} the vectors tend to be aligned
along flow lines and to smoothly change with position.

Aside from the above examples,
there are also general reasons which support 
the hypothesis of continuity.
We have already pointed out that a CLV evolves with time,
according to Eqs.~(\ref{eq:xp}) and (\ref{eq:dxp}), into other
CLVs with the same LE. Assuming that the field $\vect{F}$ has
a smooth behaviour, this shows that one can define the
vector field $\vect{v}$  in such a way that it is continuous and differentiable
at least along individual trajectories. Moreover,
the numerical calculation of
the CLVs requires that also their dependence on the position $\vect{ x}$ 
in the phase space is continuous, at least in some domain
$D$. Indeed, numerical calculations are always based on approximation
of real numbers with truncated binary representations: in order to
be meaningful, 
the represented relations must be at least continuous.

Rigorous results on continuity and differentiability are available in
a related field: the differentiability of Anosov splitting has been
extensively studied in the context of geodesic flows. In some of these
studies, differentiability of class $\mathcal{C}^{\infty}$ 
or even $\mathcal{C}^2$  is declared to be a rare
property~\cite{fang2005, benoist1992}, due to the connection with a
quite strict necessary and sufficient condition~\cite{kanai1988,
  benoist1992}. A sufficient condition for having class
$\mathcal{C}^1$ is also known~\cite{benoist1992, hirsch1969,
  hirsch1975}, but no necessary conditions. Summarizing, theorems
prove the differentiability only in a few cases, but the continuity is
usually assumed to hold.

We formalize the above considerations by introducing the following
definition.

\begin{definition}[{\bf covariant Lyapunov field}]\label{def:lvf}
  Let $\vect{F}(\vect{x})$ be a vector field on the Riemannian
  manifold $X$, and let $D\subseteq X$ be an invariant set with
  respect to the evolution generated by $\vect{F}$. Let
  $\vect{v}(\vect{x})$ be a continuous vector field on $D$, such that 
the function $\lVert \vect{v}(\vect{x})\rVert$ is differentiable.
  If, at every point $\vect{x}\in D$, $\vect{v}(\vect{x})$ is a CLV
  with LE $\lambda$ independent of $\vect{x}$, then we say that
  $\vect{v}(\vect{x})$ is a ``covariant Lyapunov field" (CLF) on the
  domain $D$ with LE $\lambda$.
\end{definition}

Since CLVs are defined up to a multiplicative factor, 
it is clear that, given any vector field $\vect{v}(\vect{x})$
satisfying the above definition, it is always possible to
redefine it in such a way that one has everywhere
$\lVert \vect{v}(\vect{x})\rVert = 1$. 
It might then seem reasonable
to include such a condition directly into
 definition \ref{def:lvf}, thus automatically ensuring that 
the function $\lVert \vect{v}(\vect{x})\rVert$ is differentiable.
One has however to keep in mind that a
generic phase space is not equipped with 
any intrinsic metric, and the value of LEs is independent of
the norm used in Eq.~(\ref{eq:limits}). Thus there is no basis
for considering CLFs only those vector fields
which are normalized with respect to a particular norm.

Rigorously speaking, 
if $\nu$ is the multiplicity of the LE $\lambda$, one may
arbitrarily choose, at each point of $D$, $\nu$ CLVs which form a
basis of the corresponding linear subspace. 
We then see that a suitable choice of $\nu$ linearly independent CLVs,
at each point of $D$, determines $\nu$ linearly independent CLFs
associated with the LE $\lambda$.

\section{A global equation characterizing the covariant Lyapunov fields}

The aim of this section is to establish an important characteristic property
of CLFs, namely that of being 
the global solutions of a particular differential equation on their domain. 
The following lemma deals with a differential
equation which, owing to its formal analogy with Eq.~(\ref{eq:dxp}),
will be an essential tool for reaching this goal.  An elementary but
important property of this equation, which is highlighted in the
lemma, is the existence of a simple transformation relating different
solutions with one another.

\begin{lemma}\label{lemma0}
Let $\vect{F}(\vect{x})$ be a differentiable vector field 
on a manifold $X$, and let $\gamma$
be a trajectory defined by a
  function $\vect{x}(t)$ satisfying Eq.~(\ref{eq:xp}) for
  $-\infty<t<+\infty$.
Let us suppose that $\vect v(t)$ and $b(t)$
are respectively a vector and a scalar function satisfying on the
trajectory $\gamma$ the differential equation
 \begin{equation}
    \frac{\mathrm{d}}{\mathrm{d}t} v^{\mu}(t) =
    \partial_\nu F^{\mu} [\vect x(t)]v^{\nu}(t) -   b(t) v^{\mu}(t)\,.
    \label{eq:traject0}
  \end{equation}  
Then, for any arbitrary nonvanishing smooth scalar function
$a(t)$, the vector function
\begin{equation}\label{eq:wgt}
\vect{v}'(t)=a(t)\vect{v}(t) 
\end{equation}
satisfies the equation
\begin{equation}
    \frac{\mathrm{d}}{\mathrm{d}t} v'^{\mu}(t) =
    \partial_\nu F^{\mu}[\vect x(t)] v'^{\nu}(t) -   b' (t)v'^{\mu}(t)
    \label{eq:traject1}
  \end{equation}  
with 
\begin{equation}\label{eq:cgt}
  b'(t)=b(t) - \frac{\mathrm{d}}{\mathrm{d} t} \ln\left|a(t)\right| 
  \, .
\end{equation}
\end{lemma}
\begin{proof}
From Eqs.~(\ref{eq:traject0}) and (\ref{eq:wgt}) it follows that
\begin{align*}
\frac{\mathrm{d}}{\mathrm{d}t} v'^{\mu}(t) &=
a(t)\frac{\mathrm{d}}{\mathrm{d}t} v^{\mu}(t)+
\frac{\mathrm{d}}{\mathrm{d}t}a(t) v^{\mu}(t) \\
&=\partial_\nu F^{\mu}[\vect x(t)] v'^{\nu}(t) -   b (t)v'^{\mu}(t)+
\frac{\mathrm{d}}{\mathrm{d}t} a(t) \frac {v'^\mu(t)}{a(t)} 
\end{align*}
from which Eqs.~(\ref{eq:traject1})--(\ref{eq:cgt}) are obtained.
\end{proof}

The following lemma shows that any CLF, associated with a
nondegenerate LE on an invariant domain $D$,
is the solution of a particular differential
equation of first order along any trajectory contained in $D$.

\begin{lemma}\label{lemma:eq_br}
  Let the vector field $\vect{F}(\vect{x})$ generate a flow on the
  Riemannian manifold $X$ according to Eq.~(\ref{eq:xp}), and let
  $\vect{v}(\vect{x})$ be a CLF on an invariant domain $D\subseteq
  X$, corresponding to a nondegenerate LE $\lambda$.  
Then, for any trajectory $\gamma\subseteq D$, defined by a
  function $\vect{x}(t)$ satisfying Eq.~(\ref{eq:xp}) for
  $-\infty<t<+\infty$, 
the function
  $\vect{v}\left[\vect{x}(t)\right]$ is differentiable with respect to
  the time $t$, and there exists on $D$ a scalar function 
$b(\vect{x})$  such that the differential equation
\begin{equation}
    \frac{\mathrm{d}}{\mathrm{d}t} v^{\mu}[\vect{x}(t)] =
    \left[\partial_\nu F^{\mu} v^{\nu} -    
b v^{\mu}\right]_{\vect{x}(t)}
    \label{eq:traject}
  \end{equation}  
  holds on $\gamma$. 
If $\lVert \vect{v}[\vect{x}(t)]\rVert$ and
$\lVert \vect{v}[\vect{x}(t)]\rVert^{-1}$ are both limited
on $\gamma$,
then the time average of $b$ over $\gamma$ is
  equal to the LE $\lambda$:
  \begin{equation}\label{eq:time_av}
    \lim_{t\to \pm\infty}\frac{\int_0^t b\left[\vect{x}(t')\right]{\mathrm{d}}t'}{t}
    =\lambda \, .
  \end{equation}
\end{lemma}

\begin{proof}
  For a given trajectory $\gamma\subseteq D$, let $\vect{\delta x}(t)$
  be the solution of Eq.~(\ref{eq:dxp}) with initial condition
  $\vect{\delta x}(0)=\vect{v}(\vect{x}_0)$, where
  $\vect{x}_0=\vect{x}(0)$.  As we have already recalled, the property
  of being a CLV is maintained by the linearized flow, so
  $\vect{\delta x}(t)$ is for any $t$ a CLV with LE $\lambda$ at the
  point $\vect{x}(t)$.  On the other hand, by hypothesis also
  $\vect{v}\left[\vect{x}(t)\right]$ is for any $t$ a CLV with LE
  $\lambda$.  Since $\lambda$ is assumed to be a nondegenerate LE, the
  corresponding CLVs form at all the points of $D$ a 1-dimensional
  linear space, so there exists a nonvanishing real function $a(t)$
  such that
  \begin{equation}\label{eq:delta_x}
\vect{v}\left[\vect{x}(t)\right]=   a(t) \vect{\delta x}(t)
  \end{equation}
  at all the points of $\gamma$. Since $\vect{v}\left[\vect{x}(t)\right]$ and
    $\vect{\delta x}(t)$ are continuous 
  functions of $t$ and  $a(0)=1$, we get $a(t)>0$ for any $t$.
  For this reason, taking the norm of of Eq.~(\ref{eq:delta_x}) gives
  \begin{equation}\label{eq:delta_xn}
  \lVert \vect{v}\left[\vect{x}(t)\right] \rVert  
    =
    a(t) \lVert \vect{\delta x}(t) \rVert
 \, .
  \end{equation}
  Since both $\lVert \vect{\delta x}(t) \rVert$ and $\lVert
  \vect{v}\left[\vect{x}(t)\right] \rVert$ are differentiable
  functions of $t$, the above equation implies that also $a(t)$ is differentiable,
so from Eq.~(\ref{eq:delta_x}) we obtain that  $\vect{v}\left[\vect{x}(t)\right]$ 
is differentiable with respect to  the time $t$.

Eq.~(\ref{eq:dxp}) has the same form as Eq.~(\ref{eq:traject0}) with $b=0$.
Hence, by applying lemma \ref{lemma0}, we obtain from Eq.~(\ref{eq:delta_x}) 
  that Eq.~(\ref{eq:traject}) holds with
  \begin{equation}\label{eq:b}
    b\left[\vect{x}(t)\right]= 
    -\frac{\mathrm{d}}{\mathrm{d}t}\ln a(t)\, .
  \end{equation}
  The fact that $\vect{v}(\vect{x}_0)$ is a CLV with LE $\lambda$
  means that
\[
\lambda =\lim_{t\to \pm\infty} \frac{\ln \lVert \vect{\delta x}(t) \rVert}{t}\,.
\]
From Eq.~(\ref{eq:delta_xn}) we get $ \lVert \vect{\delta x}(t) \rVert=
\lVert \vect{v}\left[\vect{x}(t)\right]\rVert/a(t)$ and so
     \begin{equation}\label{eq:def:lambda:theo}
    \lambda =\lim_{t\to \pm\infty}\left(
    \frac{\ln \lVert \vect{v}\left[\vect {x}(t)\right] \rVert} {t} 
-\frac{\ln a(t)}{t}\right) \, .
  \end{equation}
  From Eq.~(\ref{eq:b}) and from $a(0)=1$ we get
  \begin{equation}\label{eq:sol_cauchy}
    \ln a(t)=-\int_0^t b\left[\vect{x}(t')\right]
        {\mathrm{d}}t' \, .
  \end{equation}  
  Moreover, if  $\lVert
  \vect{v}\rVert$ and $\lVert \vect{v}\rVert^{-1}$ are both limited on
  $\gamma$, then
  \[
  \lim_{t\to \pm\infty}\frac{\ln \lVert \vect{v}\left[\vect {x}(t)\right] \rVert}{t}=0\,,
  \]
so from Eq.~(\ref{eq:def:lambda:theo}) one obtains Eq.~(\ref{eq:time_av}).
\end{proof}

Note that Eq.~(\ref{eq:traject}) is formally similar to Eq.~(\ref{eq:dxp}), but it
includes an additional term involving a scalar function $b$.
As we have recalled in the preceding section,
the norm of a CLV evolving according to the tangent dynamics, 
i.e.\ as the displacement vector $\vect{\delta x}$ in Eq.~(\ref{eq:dxp}),
would increase (resp.\ decrease) exponentially with time if
the corresponding LE is positive (resp.\ negative). 
The term in the Eq.~(\ref{eq:traject}) containing
the scalar function $b$ has just the effect of compensating this increase (or decrease)
and making the time evolution compatible with the existence of a vector field
having a bounded norm everywhere. This exact compensation is the reason why
the time average of $b$, when the CLF has a bounded norm, just equals the value
of the LE $\lambda$, as shown by Eq.~(\ref{eq:time_av}). 

Our goal is now to exploit lemma \ref{lemma:eq_br}
to derive a global differential equation which characterizes CLFs for
ergodic systems. As a first step in this direction, we note that
Eq.~(\ref{eq:traject}) can be written in a more compact form by
making use of the concept of Lie derivative. If $\vect{v}$ is a
differentiable vector field as well as $\vect{F}$, then it is well known
that the Lie derivative $\mathcal{L}_{\vect{F}}\vect{v}$ 
of $\vect{v}$ with
respect to $\vect{F}$ is equal to the commutator of the two fields:
\begin{equation} \label{eq:alt:def:lie:derivative}
 \left( \mathcal{L}_{\vect{F}} \vect{v}\right)^{\mu}  =
\left[\vect{F},\vect{v}\right]^\mu =
  F^{\nu} \partial_{\nu} v^{\mu} - v^{\nu} \partial_\nu F^{\mu} \, .
\end{equation}
The first term on the right-hand-side of
Eq.~(\ref{eq:alt:def:lie:derivative}) represents the total derivative of
$\vect{v}$ with respect to time along a field line $\vect{x}(t)$ of 
$\vect{F}$ defined by Eq.~(\ref{eq:xp}):
\[
\left. F^{\nu}  \partial_{\nu} v^{\mu}\right |_{\vect{x}(t)}
=\frac{\mathrm{d}}{\mathrm{d}t} x^{\nu}(t)
\partial_{\nu} v^{\mu}\left[\vect{x}(t)\right] 
=\frac{\mathrm{d}}{\mathrm{d} t}
  v^{\mu}\left[\vect{x}(t) \right] \,.
\]
Hence Eq.~(\ref{eq:alt:def:lie:derivative}) can be rewritten as
\begin{equation} \label{eq:def:lie:derivative}
  \left(\mathcal{L}_{\vect{F}} \vect{v}\right)^{\mu} \left[\vect{x}(t)\right] =
  \frac{\mathrm{d}}{\mathrm{d} t}
  v^{\mu}\left[\vect{x}(t) \right] -
  \left. v^{\nu} \partial_\nu F^{\mu}\right |_{\vect{x}(t)} \, .
\end{equation}
This shows that the existence of the Lie derivative of $\vect{v}$, with respect
to the differentiable vector field $\vect{F}$, does not actually require the full
differentiability of $\vect{v}(\vect x)$ as a function of $\vect x$, but only the
differentiability of $\vect{v}$ along individual trajectories of $\vect{F}$.
It follows from Eq.~(\ref{eq:def:lie:derivative}) that
Eq.~(\ref{eq:traject}) can be rewritten as
 \begin{equation}
    \mathcal{L}_{\vect{F}} \vect{v} + b \vect v = 0 \,.
    \label{eq:alt:def}
  \end{equation}  
According to Eq.~(\ref{eq:alt:def:lie:derivative}), if
the field $\vect{v}$ is differentiable with respect to every coordinate,
then Eq.~(\ref{eq:alt:def}) becomes
\begin{equation} \label{eq:alt:def:commutator}
  \left[\vect{v}, \vect{F}\right] = b \vect v\,.
\end{equation} 

The fact that CLVs are determined up to an arbitrary scalar factor
implies that, if $\vect{v}(\vect{x})$ is a vector field satisfying the
hypotheses of lemma \ref{lemma:eq_br}, then the same is true
also for the vector field 
\begin{equation}\label{eq:vgt}
\vect{v}'(\vect{x})=a(\vect{x})\vect{v}(\vect{x}) \, ,
\end{equation}
where $a(\vect{x})$ is an arbitrary nonvanishing smooth scalar function.
Hence the thesis of the lemma must apply equally well to the CLF
$\vect{v}'$. In fact, it follows from lemma \ref{lemma0} that, if
$\vect{v}$ satisfies Eq.~(\ref{eq:alt:def}) on a trajectory $\gamma$,  
then $\vect{v}'$ satisfies the equation
\begin{equation}\label{eq:galt:def}
\mathcal{L}_{\vect{F}}\vect{v}' + b' \vect v' = 0
\end{equation}
with 
\begin{equation}\label{eq:bgt}
  b'\left[\vect{x}(t)\right]
=b\left[\vect{x}(t)\right] -\frac{\mathrm{d}}{\mathrm{d} t} 
\ln\left|a\left[\vect{x}(t)\right]\right|   \, .
\end{equation}
Introducing also the Lie derivative 
\[
\mathcal{L}_{\vect{F}}\phi=\frac{\mathrm{d}}{\mathrm{d} t}\phi
=F^\mu \partial_\mu \phi
\]
of a scalar function $\phi$ with respect to the vector field $\vect{F}$,
Eq.~(\ref{eq:bgt}) becomes
\begin{equation}\label{eq:bgt2}
b'=b -\mathcal{L}_{\vect{F}} \ln\left|a\right|\,.
\end{equation}
Furthermore, if 
$|a(\vect{x})|$ and $|a(\vect{x})|^{-1}$ are both limited on $D$, then
\begin{align*}
 &\lim_{t\to \pm\infty}\frac{\int_0^t b'\left[\vect{x}(t')\right]{\mathrm{d}}t'}{t}\\
    =\ &\lim_{t\to \pm\infty}\frac{\int_0^t b\left[\vect{x}(t')\right]{\mathrm{d}}t'
-\ln\left|a\left[\vect{x}(t)\right]\right|+\ln\left|a\left[\vect{x}(0)\right]\right|}{t} \\
=\ &\lim_{t\to \pm\infty}\frac{\int_0^t b\left[\vect{x}(t')\right]{\mathrm{d}}t'}{t}
=\lambda \, ,
\end{align*}
since $\ln|a(\vect{x})|$ is a limited function on $D$.  The
equivalence between Eqs.~(\ref{eq:alt:def}) and (\ref{eq:galt:def})
shows that the differential equation for the CLFs has an important
invariance property, which we will exploit later in this paper and we
will further analyze in section \ref{sec:gauge}.

In particular, it follows from Eqs.~(\ref{eq:vgt})--(\ref{eq:bgt2}) that the
normalized vector field 
\begin{equation}\label{eq:wnorm}
\vect{w}(\vect{x})=\frac
{\vect{v}(\vect{x})}{\lVert \vect{v}\left(\vect {x}\right) \rVert}
\end{equation}
satisfies the equation 
\begin{equation}\label{eq:walt:def}
\mathcal{L}_{\vect{F}}\vect{w} + c \vect w = 0
\end{equation}
with
\begin{equation}\label{eq:cgt2}
c=b +\mathcal{L}_{\vect{F}} \ln\lVert\vect v\rVert\,.
\end{equation}

For an ergodic system, the set of LEs is the same at almost all points
of $X$ with respect to the preserved measure
$\mu$~\cite{ruelle1979}. \detail{page 282, 2.1} For such systems one
can then expect that there exist CLFs defined almost everywhere on
$X$. In order to deal with this case, we shall make use of the 
following simple lemma
which states that, if a scalar function is integrable over a 
measurable manifold, then the integral of its Lie derivative with respect
to a measure-preserving flow vanishes.

\begin{lemma}\label{lemma:lie}
Let $\mu$ be a positive measure on the manifold $X$, and
let the vector field $\vect{F}$ generate a flow on 
$X$ which preserves the measure $\mu$.
If $c$ is a differentiable scalar function defined on $X$,
which is integrable over $X$ with respect to the measure $\mu$,
then 
\[
\int_X \mathrm{d}\mu(\vect{x}) \mathcal{L}_{\vect{F}}c(\vect{x})=0\,.
\]
\end{lemma}
\begin{proof}
Let $\Phi_t(\vect{x})$ be the map which describes the evolution
of the phase space $X$ at time $t$
according to Eq.~(\ref{eq:xp}), so that
\[
\frac{\mathrm{d}}{\mathrm{d}t}\Phi_t(\vect{x})=\vect{F}
\left[(\Phi_t(\vect{x})\right] 
\]
and $\Phi_0(\vect x)=\vect x\,\forall\,\vect x\in X$. Then
\[
\mathcal{L}_{\vect{F}}c(\vect{x})=\frac{\mathrm{d}}{\mathrm{d}t}
c\left(\Phi_t(\vect {x})\right)\,.
\]
Moreover,
since for all $t$ the map $\Phi_t(\vect{x})$ is a transformation of $X$ 
which preserves the measure $\mu$, we have
\[
\int_X \mathrm{d}\mu(\vect{x})c\left(\Phi_t(\vect {x})\right) 
=\int_X \mathrm{d}\mu(\vect{x})c\left(\vect {x}\right) \,, 
\]
which means that the integral on
left-hand-side of the above equation is a constant
independent of $t$. It follows that
\[
\int_X \mathrm{d}\mu(\vect{x}) \mathcal{L}_{\vect{F}}c(\vect{x})
=\frac{\mathrm{d}}{\mathrm{d}t}
\int_X \mathrm{d}\mu(\vect{x})c\left(\Phi_t(\vect {x})\right) =0\,.
\qedhere
\]
\end{proof}

We are now ready to present the first result about the global equation
characterizing the CLFs.
The following proposition can actually be considered as an extension of the
result, which was proved in lemma \ref{lemma:eq_br} for individual
trajectories, to CLFs defined at almost all points of $X$.

\begin{proposition}  \label{prop:1}
Let $\mu$ be a positive measure on the Riemaniann manifold $X$ such
that $\mu(X)<+\infty$, and let the vector field $\vect{F}$ generate an
ergodic flow on $X$ which preserves the measure $\mu$.  Let $\vect{v}$
be a CLF, corresponding to a nondegenerate LE $\lambda$, on an
invariant domain \mbox{$D\subseteq X$} such that $\mu(D)=\mu(X)$.
Then the Lie derivative of $\vect{v}$ along $\vect{F}$ exists, and
there exists a scalar field $b$ such that the differential equation 
\[
 \mathcal{L}_{\vect{F}} \vect{v} + b \vect v = 0 
\]
holds on $D$. 
Let us also suppose that the function 
$\ln \lVert \vect{v}\left(\vect {x}\right) \rVert$ is integrable over $X$ 
with respect to the measure $\mu$. Then
  \begin{equation}\label{eq:lambda_b}
    \lambda=\left<b\right>\,, 
  \end{equation}
  where $\left<b\right>$ denotes the
  average of $b$ over the manifold $X$:
  \begin{equation}\label{eq:aver_b}
    \left<b\right>=  \frac {1}{\mu(X)}
    \int_X \mathrm{d}\mu (\vect{x})\, b(\vect{x}) \, .
  \end{equation}
\end{proposition}

\begin{proof}
If $\vect v$ is  a CLF with LE $\lambda$, then the same is true for
the vector field $\vect w$ defined by Eq.~(\ref{eq:wnorm}), and
since $\lVert \vect{w}\left(\vect {x}\right) \rVert=1$ everywhere,
it follows from lemma \ref{lemma:eq_br} that
there exists on the domain $D$
a scalar field $c$ such that Eq.~(\ref{eq:walt:def}) holds at
all points of $D$ and
\begin{equation}\label{eq:lambda0}
\lambda=
\lim_{t\to \pm\infty}\frac{\int_0^t c\left[\vect{x}(t')\right]
{\mathrm{d}}t'}{t} \,.
\end{equation}
  In addition, the ergodicity implies that the time average 
  of the function $c$ along a generic trajectory equals the average of $c$
  over the phase space, so that Eq.~(\ref{eq:lambda0}) is equivalent to
\begin{equation}\label{eq:lambda}
\lambda= \frac {1}{\mu(X)}
    \int_X \mathrm{d}\mu (\vect{x})\, c(\vect{x})
=\left<c\right>\,.
\end{equation}
Since $\vect v(\vect x)=\lVert \vect v(\vect x)\rVert \vect w(\vect x)$,
it follows from Eq.~(\ref{eq:walt:def})
that $\vect{v}(\vect{x})$ satisfies Eq.~(\ref{eq:alt:def})
with 
\begin{equation}\label{eq:bb}
b(\vect{x})=c(\vect{x})-\mathcal{L}_{\vect{F}}
\ln \lVert \vect{v}\left(\vect {x}\right) \rVert \,,
\end{equation}
in accordance with Eq.~(\ref{eq:cgt2}).
If the function $\ln \lVert \vect{v}\left(\vect {x}\right) \rVert$ is 
integrable  over $X$, by applying lemma \ref{lemma:lie}
we get $\left<b\right>=\left<c\right>$, so Eq.~(\ref{eq:lambda_b})
follows from Eq.~(\ref{eq:lambda}).
\end{proof}

The following proposition can be considered in some respect as the
inverse of the previous one. It shows in fact
that in an ergodic system, under quite
general hypotheses, the fact of satisfying Eq.~(\ref{eq:alt:def}) is a
sufficient condition for a vector field $\vect{v}$ in order to be a
CLF.

\begin{proposition}
  \label{prop:alt:def}
Let $\mu$ be a positive measure on the Riemaniann manifold $X$ such
that $\mu(X)<+\infty$, and let the vector field $\vect{F}$ generate an
ergodic flow on $X$ which preserves the measure $\mu$.  Let $\vect{v}$
and $b$ be respectively a nonvanishing vector field and a scalar
function satisfying the equation
\[
 \mathcal{L}_{\vect{F}} \vect{v} + b \vect v = 0 
\]
on an invariant domain \mbox{$D\subseteq X$}, such that
$\mu(D)=\mu(X)$. Let us suppose that the function
\begin{equation}\label{eq:bb2}
c(\vect{x})=b(\vect{x})+\mathcal{L}_{\vect{F}}
\ln \lVert \vect{v}\left(\vect {x}\right) \rVert 
\end{equation}
is integrable over $X$ with respect to the measure $\mu$, and
that $\lVert \vect{v}\left(\vect {x}\right) \rVert$ is differentiable
on $D$. Then $\vect{v}$ is a CLF with LE $\lambda=\left<c\right>$ on
an invariant domain $D'\subseteq D$ such that $\mu(D')=\mu(X)$. If, in
addition, also $\ln \lVert \vect{v}\left(\vect {x}\right) \rVert$ is
integrable over $X$, then $\left<b\right>=\left<c\right>=\lambda$.
\end{proposition}
Note that, since $\mu(X)<+\infty$, if $\lVert
\vect{v}(\vect{x})\rVert$ and $\lVert \vect{v}(\vect{x})\rVert^{-1}$
are both limited on $D$, then the hypothesis in propositions
\ref{prop:1} and \ref{prop:alt:def} about the integrability of the
function $\ln \lVert \vect{v}\left(\vect {x}\right) \rVert$ is
obviously satisfied.  Moreover, in such a case, the
  hypothesis on the integrability of the function $c\left(\vect {x}\right)$, 
defined by  Eq.~(\ref{eq:bb2}), is equivalent to the hypothesis on the
  integrability of the function $b\left(\vect {x}\right)$ appearing in
  Eq.~(\ref{eq:alt:def}).  In other words, if $\lVert
  \vect{v}(\vect{x})\rVert$ and $\lVert \vect{v}(\vect{x})\rVert^{-1}$
  are both limited on $D$ and $b\left(\vect {x}\right)$ is integrable, 
then $c\left(\vect {x}\right)$, defined by
  Eq.~(\ref{eq:bb2}) is also integrable as requested by the hypothesis
  of Prop.~\ref{prop:alt:def}. In section \ref{prop:norm}
(see proposition \ref{prop:alt:def2}) we will show that
these hypotheses take a simpler form when the space $X$ is compact.

\begin{proof}[Proof of Proposition \ref{prop:alt:def}]
If $\vect{v}(\vect{x})$ satisfies Eq.~(\ref{eq:alt:def}), then
the vector field $\vect{w}(\vect{x})$ defined by Eq.~(\ref{eq:wnorm})
satisfies Eq.~(\ref{eq:walt:def}) with $c$ given by Eq.~(\ref{eq:bb2}).
  Let us take a point $\vect{x}_0\in D$ and let $\vect{x}(t)$ be the
  corresponding trajectory, i.e.\ the solution of Eq.~(\ref{eq:xp})
  with initial condition $\vect{x}(0)=\vect{x}_0$.  It follows from
Eq.~(\ref{eq:walt:def}) that
  \begin{equation}\label{eq:trajectw}
    \frac{\mathrm{d}}{\mathrm{d}t}w^{\mu} \left[\vect{x}(t)\right]
    =\left[ w^\nu \partial_\nu F^\mu  -c w^\mu\right]_{\vect{x}(t)} \,.
  \end{equation}
If we define
\begin{equation}\label{eq:sol_cauchy1}
    a(t)=\exp\left\{\int_0^t c\left[\vect{x}(t')\right]
        {\mathrm{d}}t' \right\}\,,
  \end{equation} 
it follows from lemma \ref{lemma0} that the vector function
$\vect w'(t)=a(t)\vect{w}\left[\vect{x}(t)\right]$ satisfies the equation
\begin{equation}
 \frac{\mathrm{d}}{\mathrm{d}t}w'^{\mu} (t)
    = w'^\nu (t)\partial_\nu F^\mu\left[\vect{x}(t)\right] \,.
\end{equation}
We then see that $\vect w'(t)=\vect{\delta x}(t)$, where 
$\vect{\delta x}(t)$ is the solution of Eq.~(\ref{eq:dxp}) with initial condition 
$\vect{\delta  x}(0)=\vect{w}(\vect{x}_0)$.

Since $\vect{\delta x}(t)=a(t)\vect{w}\left[\vect{x}(t)\right]$ and
$\lVert \vect{w}\left(\vect {x}\right) \rVert=1$ everywhere,
we have
  \begin{align}
    \label{eq:timeav}
    \lim_{t\to +\infty} \frac{\ln \lVert \vect{\delta x}(t) \rVert}{t}  
    & =\lim_{t\to +\infty}\frac{\ln a(t)}{t}    \nonumber \\
    &    =\lim_{t\to +\infty}\frac{\int_0^t c\left[\vect{x}(t')\right]
{\mathrm{d}}t'}{t} \, .
  \end{align}
  The last member of Eq.~(\ref{eq:timeav}) represents the time average
  of the function $c$ on the considered trajectory for
  positive times. Since $c$ is integrable over $X$ and the
system is ergodic, for almost all the points  $\vect{x}_0\in D$ with
respect to measure $\mu$
this average equals the average of $c$ over the
manifold $X$, so 
\[
 \lim_{t\to +\infty} \frac{\ln \lVert \vect{\delta x}(t) \rVert}{t}=
\frac {1}{\mu(X)}
    \int_X \mathrm{d}\mu (\vect{x})\, c(\vect{x})
=\left<c\right>\,.
\]
  
  By analyzing in a similar way the limit for $t\to -\infty$, we
  obtain that there exists a subset $D'\subseteq D$, with
  $\mu(D')=\mu(X)$, such that
  \[
  \lim_{t\to -\infty} \frac{\ln \lVert \vect{\delta x}(t) \rVert}{t} =
  \lim_{t\to +\infty} \frac{\ln \lVert \vect{\delta x}(t) \rVert}{t}=
  \langle c\rangle
  \]
  for all the points $\vect{x}_0\in D'$.  According to definition
  \ref{def:clv}, this means that $\vect{w}(\vect{x}_0)$ is a CLV at
  $\vect{x}_0$ with LE $\lambda=\langle c\rangle$, and the same
is then true for the vector $\vect{v}(\vect{x}_0)$.  The set $D'$ is
  obviously invariant under the evolution of the system so, by
  applying definition \ref{def:lvf}, we conclude that $\vect{v}$ is a
  CLF on $D'$ with LE $\lambda=\langle c\rangle$.

Finally, if the function
$\ln \lVert \vect{v}\left(\vect {x}\right) \rVert$ is 
integrable over $X$, by applying lemma \ref{lemma:lie}
we get from Eq.~(\ref{eq:bb2})
that also $b$ is integrable over $X$ and
 $\left<b\right>=\left<c\right>=\lambda$.
\end{proof}

Propositions \ref{prop:1} and \ref{prop:alt:def} together imply the
remarkable fact that, if the system is ergodic and $\lambda$ is a
nondegenerate LE, then a vector field $\vect{v}$ is a CLF with LE
$\lambda$ if and only if it satisfies Eq.~(\ref{eq:alt:def}) almost
everywhere on $X$.  Note that this is a {\em global} condition on the
vector field $\vect{v}$.  It is in fact easy to see that a {\em local}
solution of the first order differential equation (\ref{eq:alt:def}),
for an arbitrary scalar function $b$,
can be obtained after arbitrarily assigning the vector $\vect{v}$ on a
$(n-1)$-dimensional surface $\sigma$ transversal to the flow generated by
$\vect{F}$.  This obviously means that being a local solution of
Eq.~(\ref{eq:alt:def}) does not imply that a vector field $\vect{v}$
is a CLF. 
If one tries to extend such a local solution to the whole phase space
by solving Eq.~(\ref{eq:alt:def})
along individual trajectories, one is obviously faced by the problem that
each trajectory crosses the surface $\sigma$ infinitely many times.
Assuming that at a given crossing $\vect{v}\left[\vect{x}(t)\right]$ has the 
right value which was initially
assigned on $\sigma$, the same would not in general be true
for the subsequent times at which the trajectory crosses of the surface
again. According to proposition \ref{prop:1}, on the other hand,
if the values of $\vect{v}$ assigned at all points of
$\sigma$ correspond to CLVs with a given
nondegenerate LE $\lambda$, then there exists a scalar function
$b$ on $X$ such that a global solution of Eq.~(\ref{eq:alt:def})
can be obtained, and $b$ must satisfy Eq.~(\ref{eq:lambda_b}).

\section{Discussion of the results}\label{sec:gauge}

Since the two Eqs.~(\ref{eq:alt:def}) and (\ref{eq:galt:def}) are
formally identical, we can say that Eq.~(\ref{eq:alt:def}) is
invariant under the local ``gauge transformation'' expressed by
Eqs.~(\ref{eq:vgt}) and (\ref{eq:bgt2}).
Since the function $a$ nowhere vanishes, 
by continuity it has constant sign over the
domain $D$ of the CLF. Assuming that the sign is positive, we can
write $a(\vect{x})=e^{\varphi(\vect{x})}$, where $\varphi$ is an arbitrary
smooth scalar function. The transformation
given by Eqs.~(\ref{eq:vgt}) and (\ref{eq:bgt}) then takes the form
\begin{equation}\label{lgauge}
  \left\{\begin{array}{l}
  \vect{v}(\vect{x})  \mapsto  e^{\varphi(\vect{x})}\, \vect{v}(\vect{x}) \\
  b(\vect{x})  \mapsto b(\vect{x}) - \mathcal{L}_{\vect{F}}
  \varphi(\vect{x})  \, .
  \end{array}\right.
\end{equation}

From a mathematical point of view, such a gauge invariance recalls
that of field theories in fundamental physics. For instance, in
quantum electrodynamics, the fact the wavefunction $\psi$ is defined
at each space-time point up to an arbitrary phase factor, implies that
Dirac equation
\[
\gamma^\mu\left[i\partial_\mu-eA_\mu(x)\right]\psi(x)-m\psi(x)=0
\]
is invariant under the local gauge transformation
\begin{equation}\label{egauge}
\left\{\begin{array}{l}
  \psi(x)  \mapsto  e^{ie\alpha(x)}\, \psi(x) \\
  A^\mu(x)  \mapsto A^\mu(x) -\partial_\mu \alpha(x)
  \, ,
  \end{array}\right.
\end{equation}
where $x$ stands for the four space-time coordinates and $\alpha(x)$
is an arbitrary real scalar function \cite{zuber}.
 
The analogy between Eqs.~(\ref{lgauge}) and (\ref{egauge}) is obvious.
The transformation on the four-vector potential $A^\mu$, given by
Eq.~(\ref{egauge}), does not alter the value of the physically
relevant electromagnetic field tensor $F^{\mu\nu}=\partial^{\mu}A^\nu
-\partial^{\nu}A^\mu$. In a similar way, provided that the function
$\varphi(\vect{x})$ is integrable over $X$, it follows from
lemma \ref{lemma:lie} that
the transformation on the
scalar function $b$, given by Eq.~(\ref{lgauge}), does not alter the
physically relevant value of the LE $\lambda= \langle b\rangle$.
Suppose that a metric tensor has been defined over the manifold $X$, 
e.g.\ the euclidean tensor in a given system of coordinates.
In view of the gauge invariance which we have explained above, imposing
everywhere the condition $\lVert\vect{v}\rVert=1$ would just be
one of the infinite possible ways of ``fixing the gauge''.

It is worth remarking that the definition \ref{def:clv} of
CLV and of the corresponding LE, similarly to other definitions of
Lyapunov vectors and exponents adopted in the literature, is
explicitly based on the existence of a norm of tangent vectors, as shown by
Eq.~(\ref{eq:limits}).  The same is then true also for the definition
\ref{def:lvf} of CLF.  Despite this fact, as we have already pointed
out, both the property of being a CLV, and the value of the LE, are
actually independent of the choice of a particular metric tensor on
the space $X$.
It is therefore interesting to note that propositions \ref{prop:1} and
\ref{prop:alt:def} provide the possibility of an alternative
definition of CLF, and of the corresponding LE, which does not mention
at all the existence of a norm. One could in fact define as CLF any
vector field satisfying Eq.~(\ref{eq:alt:def}), and define its LE as
$\lambda=\langle b\rangle$.  In the case of nondegenerate CLFs in
ergodic systems, under very general hypotheses, as we have shown,
such a definition would be equivalent to definition
\ref{def:lvf}. In the case of a LE with degeneracy $\nu>1$, one could
conjecture, under suitable hypotheses, the existence of $\nu$ linearly
independent vector fields, each one satisfying an equation of the form
of Eq.~(\ref{eq:alt:def}).

It has been noticed that the CLVs ``represent the proper
generalisation of the concept of eigenvectors to a context where a
different matrix is applied at each time step.'' \cite{pikovsky}
\detail{page 58} The analogy with the eigenvector problem is
particularly evident in our alternative definition of CLF based on
Eq.~(\ref{eq:alt:def}), 
i.e.\ $-\mathcal{L}_{\vect{F}} \vect{v} =b \vect v$:
the left-hand-side is a linear operator acting on
$\vect{v}$ and the right-hand-side is the $\vect v$ itself multiplied by
a scalar. However, at variance with the usual eigenvector problem, the
scalar $b$ is a field (it is a function of the position $\vect{x}$) 
and depends on the choice of the gauge, so one should consider
as the actual eigenvalue 
the average of $b(\vect{x})$ over $X$, i.e.\ the LE $\lambda$.

\section{The equation for normalized covariant Lyapunov fields}
\label{prop:norm}

We have already underlined the fact that Eq.~(\ref{eq:alt:def}),
which according to the preceding results characterizes a CLF, does not
involve any metric on the phase space $X$. In this section we want
however to show that, if a CLF is normalized with respect
to a given metric tensor $g$, then it satisfies
a particular nonlinear differential equation. From this
equation, obviously involving the metric tensor $g$, one can derive
interesting results which also apply to generic CLFs.

We recall that,
if $g$ is the metric tensor defined on the Riemaniann manifold $X$, 
then the norm of a tangent vector $\vect{v}(\vect{x})$
is defined as
\begin{equation}\label{eq:normv}
\lVert \vect{v}(\vect{x})\rVert= 
\sqrt{v^\mu(\vect{\vect{x}}) g_{\mu\nu}(\vect{x})v^\nu(\vect{x})}\,.
\end{equation}
The Lie derivative of
$g$ with respect to the vector field $\vect F$ is given by
\[
\mathcal{L}_{\vect{F}} g_{\mu\nu}=F^\lambda \partial_\lambda g_{\mu\nu}
+g_{\mu\lambda}\partial_\nu F^\lambda + g_{\lambda\nu}\partial_\mu F^\lambda
=D_\mu F_\nu + D_\nu F_\mu\,,
\]
where $F_\nu=g_{\nu\lambda}F^\lambda$ and
$D_\mu$ is the covariant derivative associated with the
metric tensor $g$. The following proposition shows that, for a normalized 
nondegenerate CLF $\vect w$, the scalar function $c$ appearing in
Eq.~(\ref{eq:walt:def}) can be explicitly expressed as a quadratic function
of $\vect w$ itself. As a result, the CLF turns out to be the solution of
a closed nonlinear differential equation.

\begin{proposition}\label{prop:blim}
Let $\vect{F}$ be a differentiable vector field on the Riemaniann manifold
$X$, and let $D\subseteq X$ be an invariant set with respect to the evolution
generated by $\vect{F}$. Let $\vect{w}$
be a CLF, corresponding to a nondegenerate LE $\lambda$, on an
invariant domain \mbox{$D\subseteq X$} such that $\mu(D)=\mu(X)$. If
\begin{equation}\label{eq:wnorm1}
\lVert \vect{w}(\vect{x})\rVert=1 \ \forall\,\vect{x}\in D\,,
\end{equation}
then $\vect w$
satisfies Eq.~(\ref{eq:walt:def}) on the domain $D$ with 
\begin{equation}\label{eq:blie}
c= \frac 12
w^\mu \left(\mathcal{L}_{\vect{F}} g_{\mu\nu}\right) w^\nu\,.
\end{equation}
\end{proposition}

\begin{proof}
Since $\vect{w}$ is a CLF,
according to proposition \ref{prop:1} there exists a scalar function
$c$ such Eq.~(\ref{eq:walt:def}) holds on $D$.
From Eq.~(\ref{eq:wnorm1}), using Eq.~(\ref{eq:walt:def}) and
applying the Leibniz rule to the calculation of the Lie derivative along a
field line of $\vect{F}$, we then get
\begin{align*}
0&=\mathcal{L}_F \lVert \vect{w}\rVert^2= 
w^\mu \left(\mathcal{L}_{\vect{F}} g_{\mu\nu}\right) w^\nu 
+\left(\mathcal{L}_{\vect{F}} w^\mu\right) g_{\mu\nu} w^\nu +
w^\mu g_{\mu\nu} \left(\mathcal{L}_{\vect{F}} w^\nu\right)\\
&= w^\mu \left(\mathcal{L}_{\vect{F}} g_{\mu\nu}\right) w^\nu- 2c
\end{align*}
from which Eq.~(\ref{eq:blie}) is obtained.
\end{proof}

From the above proposition we can derive an explicit expression, involving the metric 
tensor $g$, for the LE associated with a generic CLF.

\begin{proposition}  \label{prop:LE}
Let $\mu$ be a positive measure on the Riemaniann manifold $X$ such
that $\mu(X)<+\infty$, and let the vector field $\vect{F}$ generate an
ergodic flow on $X$ which preserves the measure $\mu$.  Let $\vect{v}$
be a CLF, corresponding to a nondegenerate LE $\lambda$, on an
invariant domain \mbox{$D\subseteq X$} such that $\mu(D)=\mu(X)$.
Then
\begin{equation}\label{eq:LE}
   \lambda=  \frac {1}{\mu(X)}
    \int_X \mathrm{d}\mu\,
\frac{v^\mu \left(\mathcal{L}_{\vect{F}} g_{\mu\nu}\right) v^\nu}
{2\lVert \vect{v}\rVert^2} \, .
  \end{equation}
\end{proposition}

\begin{proof}
The vector field $\vect w(\vect x)=\vect v(\vect x)/
\lVert \vect v(\vect x)\rVert$ is a CLF with LE $\lambda$ such that
$\lVert w(\vect x)\rVert =1$ for any $\vect x$. Hence, 
according to proposition \ref{prop:1}, there exists a scalar function $c$
such that the equation $\mathcal{L}_{\vect{F}} \vect{w} + c \vect w = 0$
holds on $D$ and $\lambda=\left<c\right>$. Furthermore, according to
proposition \ref{prop:blim}
\begin{equation}\label{eq:cscalar}
c= \frac 12
w^\mu \left(\mathcal{L}_{\vect{F}} g_{\mu\nu}\right) w^\nu
=\frac{v^\mu \left(\mathcal{L}_{\vect{F}} g_{\mu\nu}\right) v^\nu}
{2\lVert \vect{v}\rVert^2}\,,
\end{equation}
so Eq.~(\ref{eq:LE}) is obtained.
\end{proof}

Note that the expression of $\lambda$ given by Eq.~(\ref{eq:LE}) is
manifestly invariant under the transformation (\ref{eq:vgt}) on the CLF.
Hence this formula expresses the LE only as a function of the direction of
the corresponding one-dimensional subspace at each point of the domain $D$.

If $a(\vect{x})$ is a quadratic form, i.e.\ a symmetric covariant tensor
of order 2, we define its norm as
\begin{equation}\label{eq:norma}
\lVert a(\vect{x}) \rVert =\sqrt{a_{\mu\nu}(\vect{x}) a_{\mu'\nu'}(\vect{x})
g^{\mu\mu'}(\vect{x})g^{\nu\nu'}(\vect{x})}\,.
\end{equation}
It is then easy to see that for any vector $\vect v(\vect x)$
\begin{equation}\label{eq:alim}
\lvert v^\mu (\vect{x}) a_{\mu\nu}(\vect{x}) v^\nu (\vect{x})\rvert
\leq \lVert a(\vect{x}) \rVert \lVert v(\vect{x}) \rVert^2\,.
\end{equation}
For a CLF $\vect w$, such that 
$\lVert \vect{w}(\vect{x})\rVert=1 \ \forall\,\vect{x}\in D$,
we thus get from Eq.~(\ref{eq:blie}) 
\begin{equation}\label{eq:blim}
\lvert c(\vect{x})\rvert \leq \frac 12 \lVert 
\mathcal{L}_{\vect{F}} g(\vect{x}) \rVert\,,
\end{equation}
and from Eq.~(\ref{eq:LE}) we obtain the following upper bound for the
absolute value of any LE $\lambda$ of the system:
\begin{equation}\label{eq:upperLE}
   |\lambda|\leq  \frac {1}{2\mu(X)}
    \int_X \mathrm{d}\mu (\vect{x})\,
\lVert \mathcal{L}_{\vect{F}} g(\vect{x}) \rVert \, .
  \end{equation}

Note that the upper bound provided by the above equation depends only
on the field $\vect F$ defining the dynamical system and not on the CLF
associated with $\lambda$. It should however be remarked that the right-hand
side of the above equation exhibits a dependence on the choice of the
metric tensor $g$, which instead to a large extent does not affect the value
of the LEs. The above equation can therefore also be interpreted as a
condition that all metrics consistent with a given set of LEs $(\lambda_1, \dots,
\lambda_n)$ must satisfy. One can also write
\[
\lVert \mathcal{L}_{\vect{F}} g(\vect{x}) \rVert^2=
-\left(\mathcal{L}_{\vect{F}} g_{\mu\nu}\right)\left(\mathcal{L}_{\vect{F}} g^{\mu\nu}\right)
=2D_\mu F_\nu \left(D^\mu F^\nu + D^\nu F^\mu\right)\,,
\]
with
\[
\mathcal{L}_{\vect{F}} g^{\mu\nu}=F^\lambda \partial_\lambda g^{\mu\nu}
-g^{\mu\lambda}\partial_\lambda F^\nu - g^{\lambda\nu}\partial_\lambda F^\mu
=-D^\mu F^\nu - D^\nu F^\mu\,.
\]
Defining $L={\rm max}\{|\lambda_1|, \dots, |\lambda_n|\}$, the condition on the metric
expressed by Eq.~(\ref{eq:upperLE}) can then be written as
\[
\frac {1}{\mu(X)}\int_X \mathrm{d}\mu (\vect{x})\,
\sqrt{-\left(\mathcal{L}_{\vect{F}} g_{\mu\nu}\right)\left(\mathcal{L}_{\vect{F}} g^{\mu\nu}\right)}
\geq 2L \,.
\]

If we use a system of coordinates such that
$g_{\mu\nu}(\vect{x})=\delta_{\mu\nu}$, where $\delta_{\mu\nu}$ is
Kroenecker's symbol, then the norms of $\vect{v}(\vect{x})$ and
$a(\vect{x})$ take the more familiar forms
\begin{align}
\lVert \vect{v}(\vect{x})\rVert &= \sqrt{
\sum_{\mu=1}^n\left[v^\mu(\vect{x})\right]^2}\,, \label{eq:norme} \\
\lVert a(\vect{x})\rVert &= \sqrt{
\sum_{\mu=1}^n \sum_{\nu=1}^n
\left[a_{\mu\nu}(\vect{x})\right]^2}\,.
\end{align}
Since the definition of CLV is to a large extent
independent of the particular metric adopted on $X$, the simplest
choice is to use the euclidean metric
in a given system of coordinates, so that $g_{\mu\nu}(\vect{x})
=\delta_{\mu\nu}$ everywhere. In that case 
the Lie derivative of the metric tensor can simply be written as
\[
\mathcal{L}_{\vect{F}} \delta_{\mu\nu}=
\partial_\mu F^\nu + \partial_\nu F^\mu \,,
\]
so Eq.~(\ref{eq:blie}) becomes
\[
c=w^\mu \partial_\mu F^\nu w^\nu
\]
and can also be derived in an elementary way 
using Eq.~(\ref{eq:trajectw}). Furthermore,
Eqs.~(\ref{eq:LE}) and (\ref{eq:upperLE}) become respectively
\begin{align}
\lambda &=  \frac {1}{\mu(X)}
    \int_X \mathrm{d}\mu \,
\frac{v^\mu  \partial_\mu F^\nu v^\nu}
{\lVert \vect{v}\rVert^2} \, , \\
|\lambda| &\leq  \frac {1}{\mu(X)}
    \int_X \mathrm{d}\mu\,
\sqrt{\frac 12 \sum_{\mu=1}^n \sum_{\nu=1}^n \partial_\mu F^\nu
\left(\partial_\mu F^\nu + \partial_\nu F^\mu\right)}\, .
\end{align}

Thanks to the results obtained in this section
we can formulate a proposition showing that, when $X$
is compact, the result provided by proposition \ref{prop:alt:def}
can be obtained under simplified hypotheses on $\vect{v}$ and $b$.
We recall that,
for hamiltonian systems, $X$ can be identified with a level surface of
the hamiltonian function $H$, which is typically a compact set.

\begin{proposition}
  \label{prop:alt:def2}
Let $\mu$ be a positive measure on a compact Riemaniann manifold $X$ such
that $\mu(X)<+\infty$, and let the vector field $\vect{F}$ generate an
ergodic flow on $X$ which preserves the measure $\mu$.  Let $\vect{v}(\vect {x})$
and $b(\vect {x})$ be respectively a nonvanishing vector field and a scalar
function satisfying the equation
\[
 \mathcal{L}_{\vect{F}} \vect{v} + b \vect v = 0 
\]
on an invariant domain \mbox{$D\subseteq X$}, such that
$\mu(D)=\mu(X)$. Let us also suppose that 
$\lVert \vect{v}\left(\vect {x}\right) \rVert$ is differentiable on $D$. 
Then the function $c(\vect{x})$ defined
by Eq.~(\ref{eq:bb2})
is integrable over $X$ with respect to the measure $\mu$, and
$\vect{v}$ is a CLF with LE $\lambda=\left<c\right>$ on
an invariant domain $D'\subseteq D$ such that $\mu(D')=\mu(X)$. If, in
addition, $\ln \lVert \vect{v}\left(\vect {x}\right) \rVert$ is
integrable over $X$, then $\left<b\right>=\left<c\right>=\lambda$.
\end{proposition}
\begin{proof}
If $\vect{v}(\vect{x})$ satisfies Eq.~(\ref{eq:alt:def}), then
the vector field $\vect{w}(\vect{x})$ defined by Eq.~(\ref{eq:wnorm})
satisfies Eq.~(\ref{eq:walt:def}) with $c$ given by Eq.~(\ref{eq:bb2}).
Since $\lVert \vect{w}(\vect{x})\rVert=1\ \forall\,\vect{x}\in D$, 
by applying proposition \ref{prop:blim} one obtains Eq.~(\ref{eq:blim}).
The right-hand side of this equation is a continuous function
defined on the whole compact manifold $X$, and is therefore limited on $X$.
Hence $|c|$ is limited on $D$, and since $\mu(D)=\mu(X)<+\infty$,
from this it follows that $c$ is integrable over $D$. The thesis then
follows from proposition \ref{prop:alt:def}.
\end{proof}

\section{Geometrical interpretation}

Propositions \ref{prop:1} and \ref{prop:alt:def}
above only assume that $\vect{v}$ is
  continuous. As discussed above, the continuity is expected to be a
  common property, while $\vect{v}$ is likely to be differentiable only in
  special cases. It is however interesting to consider such special
  cases, because it is possible to give a geometrical interpretation
  of our alternative definition of CLF. First of all,
  Eq.~(\ref{eq:alt:def}) can be rewritten in terms of the commutator
and becomes Eq.~(\ref{eq:alt:def:commutator}).
We see that the commutator of $\vect{v}$ and $\vect{F}$ is a
linear combination of them (actually, just one of them,
$\vect{v}$). This property is called ``involutivity'' \cite{abraham}.
It is also well-known that, under very general hypotheses, 
the vector field $\vect{F}$ itself is a CLF with LE $\lambda=0$.
Since \mbox{$\mathcal{L}_{\vect{F}}\vect{F}=[\vect{F},\vect{F}]=0$},
this result can also be deduced from proposition \ref{prop:alt:def}
under the hypothesis that the function 
$\ln \lVert \vect{F}\left(\vect {x}\right) \rVert$ is
integrable over $X$.
\detail{involutive distribution, chapter 4.3, page 298; it is also
  called integrable distribution, definition 2.2.25, page 92;
  Frobenius theorem 2.2.26, page 93} The involutivity of the couple
$(\vect{v}$, $\vect{F})$ for every CLF $\vect{v}$ allows us to apply the
Frobenius theorem~\cite{abraham}\detail{ see integrable
  distribution, definition 2.2.25, page 92; Frobenius theorem 2.2.26,
  page 93}: given any CLF $\vect{v}$, the subbundle of the tangent
bundle spanned by $\vect{v}$ and $\vect{F}$ arises from a (local)
regular foliation. This concept is expressed by the following
proposition. 

\begin{proposition} \label{prop:geo:CLF:foliation}
  Let the vector field $\vect{F}(\vect{x})$ generate a flow on the
  Riemannian manifold $X$ according to Eq.~(\ref{eq:xp}), and let
  $\vect{v}(\vect{x})$ be a differentiable CLF, linearly independent
of $\vect F$, on an invariant domain $D\subseteq X$. 
Then the couple $(\vect{v}$, $\vect{F})$ generates a
  regular foliation of $D$. Each leave of the foliation contains whole
  trajectories.
\end{proposition}

\begin{proof}
  The existence of the regular foliation is ensured by the
  involutivity of the couple $(\vect{v}$, $\vect{F})$, thanks to 
  Frobenius theorem. Since one of the generators 
of the foliation is $\vect{F}$, the
  leaves contain whole orbits generated by $\vect{F}$.
\end{proof}

In the context of Anosov flows, \cite{benoist1992} it is usual to
identify the central stable and unstable manifolds; they are tangent
to all the CLF with negative and positive LE, respectively.  These
manifolds contain the leaves of the above-mentioned foliations
generated by each $\vect{v}$ and $\vect{F}$ and are known to be
of class $\mathcal{C}^{\infty}$.

We now want to show that it is possible to derive a result which
in some sense is the inverse of
that expressed by proposition \ref{prop:geo:CLF:foliation}. 

\begin{proposition} \label{prop:geo:foliation:CLF}
Let $\mu$ be a positive measure on the Riemaniann manifold $X$
such that $\mu(X)<+\infty$, and
let the vector field $\vect{F}$ generate, according to Eq.~(\ref{eq:xp}), 
an ergodic flow on $X$ which preserves the measure $\mu$.
Let $\vect{F}$ be a CLF with LE $\lambda=0$ and be such that
the function $\lVert\mathcal{L}_{\vect{F}} g(\vect{x})\rVert$
is integrable over $X$ with respect to the measure $\mu$.
Let also
  $\vect{v}_i(\vect{x})$, for $i=1, \dots, n-1$, 
be $n-1$ additional CLFs, on an invariant domain
  $D\subseteq X$, corresponding to nondegenerate LEs $\lambda_i$.
  If a 2-dimensional foliation $\Phi$ of $D$ is such that each leave contains whole
  trajectories of $\vect{F}$, then there exists one index $\bar i$,
with $1\leq \bar i\leq n-1$, such that the foliation is generated by 
the couple $(\vect{F},\vect{v}_{\bar i})$.
\end{proposition}

\begin{proof}
At each point $\vect{x}\in D$ of a 2-dimensional leave 
of $\Phi$ the vector $\vect{F}(\vect{x})$ is tangent to the leave, 
since the leave contains whole trajectories. We can then take,
as a  basis of the tangent space of the leave, the vector
$\vect{F}(\vect{x})$ and a vector $\vect{W}(\vect{x})$ which can be
  expressed as a linear combination of the $n-1$ CLVs 
$\vect{v}_1(\vect{x})$, $\dots$,  $\vect{v}_{n-1}(\vect{x})$.
If we also impose the condition
$\lVert \vect{W}(\vect{x})\rVert =1$, then the vector $\vect{W}
(\vect{x})$ is univocally determined (up to the sign)
at each $\vect{x}\in X$, and we
obtain in this way a vector field on $D$:
  \begin{equation} \label{eq:geo:def:W}
    \vect{W}(\vect{x}) = \sum_{i=1}^{n-1} c_i (\vect{x})
\vect{v}_i (\vect{x})\,.
  \end{equation}
  
Since $\Phi$ is a 2-dimensional foliation, 
according to Frobenius theorem $\vect{F}$ and $\vect{W}$ must be
  two involutive vector fields, thus the Lie derivative
  $\mathcal{L}_{\vect{F}}\vect{W}$ must be a linear combination of
  $\vect{F}$ and $\vect{W}$. This means that
  \begin{equation}\label{eq:lie_fw}
    \mathcal{L}_{\vect{F}}\vect{W} = \alpha \vect{F} - \beta \vect{W}\,,
  \end{equation}
  where $\alpha$ and $\beta$ are two scalar fields, and
using Eq.~(\ref{eq:geo:def:W}) to express $\vect{W}$ we get
  \begin{equation}\label{eq:prop_geom}
    \sum_{i=1}^{n-1} \mathcal{L}_{\vect{F}}(c_i \vect{v}_i) =
    \alpha \vect{F} + \beta \sum_{i=1}^{n-1} c_i \vect{v}_i\,.
  \end{equation}
It follows from proposition \ref{prop:1} that for any $i=1, \dots, n-1$
there exists a scalar function $b_i$ such that the equation
 \begin{equation}
    \mathcal{L}_{\vect{F}} \vect{v}_i + b_i \vect v_i = 0 
    \label{eq:alt:defi}
  \end{equation}  
holds on $D$. From Eq.~(\ref{eq:prop_geom}) we then get
  \begin{equation}\label{eq:foliation_c}
\alpha \vect{F} -\sum_{i=1}^{n-1}\left[\mathcal{L}_{\vect{F}}c_i-
(b_i-\beta)c_i \right] \vect{v}_i = 0    \,.
  \end{equation}
Since the set of vectors $\left(\vect{F},\vect{v}_1, \dots, 
\vect{v}_{n-1}\right)$ is
linearly independent at all points $\vect{x}\in D$, 
from the above equation we get that for all $\vect{x} \in D$
\begin{align}
\alpha(\vect{x})&=0 \,, \label{eq:alpha}\\
 \mathcal{L}_{\vect{F}}c_i(\vect{x})&=\left[b_i(\vect{x})-
\beta(\vect{x})\right]c_i(\vect{x})\ \forall\,i=1,\dots, n-1\,. \label{eq:coeffci}
\end{align}
According to Eq.~(\ref{eq:alpha}) we can rewrite Eq.~(\ref{eq:lie_fw}) as
 \begin{equation}\label{eq:lie_fw1}
    \mathcal{L}_{\vect{F}}\vect{W} +\beta \vect{W}=0\,,
  \end{equation}
which has the same form as Eq.~(\ref{eq:alt:def}) with $b=\beta$.
Since $\lVert \vect{W}(\vect{x})\rVert =1$, by applying proposition 
\ref{prop:blim} we obtain
\begin{equation}
\lvert \beta(\vect{x})\rvert \leq \frac 12 \lVert 
\mathcal{L}_{\vect{F}} g(\vect{x}) \rVert
\ \forall\, \vect{x}\in X\,.
\end{equation}
Since by hypothesis the function on the right-hand side
is integrable over $X$, the above equation
implies that also $\beta$ is integrable. We can then apply proposition
\ref{prop:alt:def} and deduce from Eq.~(\ref{eq:lie_fw1}) that
$\vect{W}$ is a CLF with LE $\lambda=\left<\beta\right>$. But since 
$\left(\vect{v}_1, \dots, \vect{v}_{n-1}\right)$ is a
set of nondegenerate CLFs, it follows from Eq.~(\ref{eq:geo:def:W})
that there must be only one index $\bar i$,
with $1\leq \bar i\leq n-1$, such that the function $c_{\bar i} (\vect{x})$
is not identically 0. Hence 
\begin{equation} \label{eq:foliation_w}
    \vect{W}(\vect{x}) = c_{\bar i} (\vect{x})
\vect{v}_{\bar i} (\vect{x})
  \end{equation}
and $\left<\beta\right>=\lambda_{\bar i}$. 

Since by construction the couple $(\vect{F},\vect{W})$ generates the
foliation $\Phi$, it follows from Eq.~(\ref{eq:foliation_w}) that
also the couple $(\vect{F},\vect{v}_{\bar i})$ generates $\Phi$.
\end{proof}

It is worth remarking that,
since $c_{\bar i} (\vect{x})\neq 0\,\forall\,\vect x\in D$, for $i=\bar i$
Eq.~(\ref{eq:coeffci}) provides
 \begin{equation}\label{eq:foliation_ci}
  \mathcal{L}_{\vect{F}}\ln |c_{\bar i}|=b_{\bar i}-\beta\,.
  \end{equation}
Since $\lVert \vect{W}(\vect{x})\rVert =1$,
Eq.~(\ref{eq:foliation_w}) implies $\ln|c_{\bar i}(\vect{x})|=
-\ln \lVert \vect{v}_{\bar i}(\vect{x})\rVert$. Therefore, 
if the function $\ln \lVert \vect{v}_{\bar i}(\vect{x})\rVert$ is integrable 
over $X$, then the same is true for the function $\ln|c_{\bar i}(\vect{x})|$.
According to proposition \ref{prop:1}, in that case
$\langle b_{\bar i}\rangle=\lambda_{\bar i}$, so 
$\langle b_{\bar i}-\beta\rangle=0$. From
Eq.~(\ref{eq:foliation_ci}) we thus get
\[
\int_X \mathrm{d}\mu(\vect{x}) \mathcal{L}_{\vect{F}}
\ln |c_{\bar i}(\vect{x})|=0\,,
\]
as required by lemma \ref{lemma:lie}.

If the space $X$ is compact, proposition \ref{prop:geo:foliation:CLF}
assumes the following simpler form.

\begin{proposition} \label{prop:geo:foliation:CLF2}
Let $\mu$ be a positive measure on the compact Riemaniann manifold $X$
such that $\mu(X)<+\infty$, and
let the vector field $\vect{F}$ generate, according to Eq.~(\ref{eq:xp}), 
an ergodic flow on $X$ which preserves the measure $\mu$.
Let $F$ be a CLF with LE $\lambda=0$ and let
  $\vect{v}_i(\vect{x})$, for $i=1, \dots, n-1$, 
be $n-1$ additional CLFs, on an invariant domain
  $D\subseteq X$, corresponding to nondegenerate LEs $\lambda_i$.
  If a 2-dimensional foliation $\Phi$ of $D$ is such that each leave contains whole
  trajectories of $\vect{F}$, then there exists one index $\bar i$,
with $1\leq \bar i\leq n-1$, such that the foliation is generated by 
the couple $(\vect{F},\vect{v}_{\bar i})$.
\end{proposition}

\begin{proof}
Since the function $\lVert \mathcal{L}_{\vect{F}} g(\vect{x}) \rVert$ is
continuous on the compact space $X$, it is limited on $X$ and
therefore, since $\mu(X)<+\infty$, integrable over $X$. The thesis
then follows from proposition \ref{prop:geo:foliation:CLF}.
\end{proof}

The two propositions \ref{prop:geo:CLF:foliation} and
\ref{prop:geo:foliation:CLF} show that the regular 
2-dimensional foliations of the
space that contain whole trajectories are just those foliations
which are
generated by one of the CLFs and $\vect{F}$. Each foliation can thus
be univocally associated with one of the non-degenerate LEs and one of
the CLFs. We suggest that the relevance
of LEs and covariant Lyapunov vectors in various fields of physics and
mathematics arises to a large extent from their connection with such
foliations, which represent an underlying fundamental geometrical
structure that characterizes any dynamical system.
 
\section{An example of application of the differential equation for the CLFs}
\label{sect:hyperbolic}

In this section we provide an example of application of our results,
showing that Eq.~(\ref{eq:alt:def:commutator}) allows us to find the CLVs
of a given flow.

We analyse the same Anosov flow considered for generating
Fig.~\ref{fig:hyperbolic}.  It is a case of the Hadamard-Gutzwiller model,
namely a geodesic flow on a
genus-2 hyperbolic surface of constant negative curvature~\cite{braun2002}. 
We start from the Poincar\'e disc, i.e.\ the unit disc $|z|\leq 1$ in
the complex plane endowed with the metric:
\begin{equation}\label{eq:pmetric}
  \mathrm{d}s = \frac{2}{1-\left|z\right|^2} |\mathrm{d}z|
\end{equation}

As dynamic variables, we consider the complex coordinate $z$ of the point
and the angle $\vartheta$ between the tangent to the geodesic and
the real axis:
\begin{equation}
  \vect{x} = \left[ \begin{array}{c} z \\ \vartheta \end{array} \right]
\end{equation}
The evolution equation of the coordinates along a geodesic 
is given by Eq.~(\ref{eq:xp}) with
\begin{equation}
  \vect{F} = \left[ \begin{array}{c}
  \frac{1-\left|z\right|^2}{2} e^{i \vartheta}
  \\
  \Im{\left(z e^{-i \vartheta}\right)}
  \end{array} \right]
\end{equation}

The Poincar\'e disc is not compact in the metric (\ref{eq:pmetric}). 
We make the domain compact by
cutting a regular hyperbolic octagon inside the disc and by identifying
its opposite sides. In the resulting manifold, which has genus 2 and preserves the
constant negative curvature, the geodesic flow is ergodic and
mixing. Details on this operation can be found e.g. in
Ref.~\cite{braun2002}.

A straightforward calculation shows that the vector field
\begin{equation}
  \vect{v_+} = \left[ \begin{array}{c}
      i \frac{1-\left|z\right|^2}{2} e^{i \vartheta}
      \\
      1 - \Re{\left(z e^{-i \vartheta}\right)}
  \end{array} \right]
\end{equation}
satisfies Eq.~(\ref{eq:alt:def:commutator}) with $b=1$. 
It can be checked that this vector field also matches
the continuity conditions on the sides of the regular
octagon. According to proposition \ref{prop:alt:def}, we can conclude
from this that $\vect{v_+}$ is a CLF with LE $\lambda =1$.
Indeed, this vector field actually corresponds to the CLF shown in
Fig.~\ref{fig:hyperbolic}.
A second solution of Eq.~(\ref{eq:alt:def:commutator}) is
\begin{equation}
  \vect{v_-} = \left[ \begin{array}{c}
      -i \frac{1-\left|z\right|^2}{2} e^{i \vartheta}
      \\
      1 + \Re{\left(z e^{-i \vartheta}\right)}
  \end{array} \right]
\end{equation}
with $b=-1$, hence $\vect{v_-}$ is a CLF
associated to the LE $\lambda=-1$. The third LE, $\lambda=0$,
is trivially associated to $\vect{v_0}=\vect{F}$.

It can be noticed that, in this case, the normalization chosen for the
CLV is such that $b$ results constant and equal to the LE $\lambda$
everywhere.

\section{Conclusion}

The concept of CLF, which we have introduced in the present paper,
sheds new light on the mathematical meaning of CLVs and on
their role for the characterization of a dynamical system. The
definition of Lyapunov vectors has historically been based on 
the asymptotic behaviour of their norm according
to the tangent dynamics of the system. On the other hand,
when CLVs are associated with a CLF,
they become the global solutions of a differential equation, 
Eq.~(\ref{eq:alt:def}). We have proved that
this remarkable property actually provides the
possibility of a new definition of the concept of CLV. This
also leads to a possible new definition of
LE, since this parameter can be considered as the average value of
the scalar function $b$ appearing in Eq.~(\ref{eq:alt:def}).

These new definitions of CLV and LE have the property that,
unlike the traditional ones, they do not rely upon the concept
of norm. The fact that the choice of the metric on the phase space
does not affect the
value of LEs and the direction of CLVs has been known for a
long time, but thanks to our new definition this feature gains
an immediate evidence. The fact that the norm of the CLVs
is undetermined is reflected in an interesting 
property of Eq.~(\ref{eq:alt:def}) which we have called ``gauge
invariance'', owing to its formal similarity to a well-known invariance
property of quantum field theories. For CLFs
which are normalized with respect to a given norm, the differential 
equation takes a special nonlinear form from which an explicit upper bound
for the absolute value of any LE can be derived.
Finally, the fact that
only global solutions of Eq.~(\ref{eq:alt:def}) represent CLFs
is  obviously related with the global
meaning of the Lyapunov exponents.

The above-mentioned results only require the continuity of CLF
over the phase space, together with their differentiability along
field lines. 
Under the additional hypothesis that CLFs are  differentiable 
along any direction
almost everywhere on the phase space, we
have proved that each CLF is in involution with the generator
$\vect{F}$ of the evolution of the dynamical system. This property
allows us to suggest a geometrical interpretation of the CLFs, 
based on Frobenius theorem. According to this interpretation,
for each dynamical system there is a set of 
2-dimensional foliations, such
that each leave contains whole trajectories. Each leave is
generated by $\vect{F}$ and one of the CLFs and is therefore
characterized by one of the Lyapunov exponents of the system.

We have provided an explicit example of application of our results
to the geodesic flow in the Hadamard-Gutzwiller model. It would be
interesting to pursue this investigation by searching for
other models in which the validity of the differential
equation for the CLFs can be directly verified, 
either analytically or numerically. Of course, our results are based
on a set of mathematical assumptions which we have specified in the
hypotheses of the various propositions that we have proved in this paper.
We are aware of the fact that these assumptions may not be valid for
certain classes of dynamical systems. As we have shown by means of a
numerical study of the H\'enon-Heiles system, there can be surfaces
in the phase space where continuity of the CLFs fails, and such surfaces must therefore
be excluded from the domain in which the differential equation for the CLFs can be
applied. Another requirements is the existence of a conserved measure, which
appears to be necessary in order to express the LE as the average on the phase space
of the scalar function $b$ appearing in the differential equation for the CLFs.
Finally, we have assumed that the phase space is finite-dimensional, and that
one can define  a metric on it such that all CLVs have finite norm.
Further investigations might be helpful in order to establish whether our results
can be at least partially extended to situations which do not fulfill all of the hypotheses
that we have assumed to hold in the present paper.

\section{Conflict of interest}
The authors have no conflicts of interest to declare. 

\section{Data availability statement}
This paper reports analytic work. Numerical results are only
reported for illustration purposes and are available
upon request.

\section{Ethics statement}
The research reported in the paper does not involve topics
raising ethical concerns.

\section{Acknowledgments}
This research received no external funding.

\bibliographystyle{unsrt}
\bibliography{lyapunov}

\end{document}